%% file: SAS-SAND.tex
\documentclass[a4paper,USenglish,cleveref, autoref, thm-restate]{lipics-v2021}

\input{preamble}
\title{Better late, then? The hardness of choosing delays to meet passenger demands in temporal graphs} 

\titlerunning{Better late, then?} 

\author{David C. {Kutner}\footnote{Corresponding author.}}{Department of Computer Science, Durham University, UK \and \url{https://dave-ck.github.io/} }{david.c.kutner@durham.ac.uk;david.kutner@glasgow.ac.uk}{https://orcid.org/0000-0003-2979-4513}{revised this work while a Research Assistant at the University of Glasgow funded by EPSRC grant EP/T004878/1, \emph{Multilayer Algorithmics to Leverage Graph Structure (MultilayerALGS)}.} 

\author{Anouk Sommer}{Karlsruher Institut für Technologie (KIT), Germany}{anouk.sommer@student.kit.edu}{https://orcid.org/0009-0006-1366-4377}{This work was partly supported by the Deutscher Akademischer Austauschdienst (DAAD) project \emph{Schnell aber spät: breaking Deutsche Bahn even more with graph theory}.}

\authorrunning{D. C. Kutner and A. Sommer} 

\Copyright{David C. Kutner and Anouk Sommer} 



\acknowledgements{The authors would like to thank the anonymous reviewers whose insightful comments led to several improvements.}

\nolinenumbers 

\EventEditors{Kitty Meeks and Christian Scheideler}
\EventNoEds{2}
\EventLongTitle{4th Symposium on Algorithmic Foundations of Dynamic Networks (SAND 2025)}
\EventShortTitle{SAND 2025}
\EventAcronym{SAND}
\EventYear{2025}
\EventDate{June 9--11, 2025}
\EventLocation{Liverpool, GB}
\EventLogo{}
\SeriesVolume{330}
\ArticleNo{7}

\begin{document}

\maketitle

\begin{abstract}
In train networks, carefully-chosen delays may be beneficial for certain passengers, who would otherwise miss some connection.
Given a simple (directed or undirected) temporal graph and a set of passengers (each specifying a starting vertex, an ending vertex, and a desired arrival time), we ask whether it is possible to delay some of the edges of the temporal graph to realize all the passengers' demands. We call this problem \textsc{DelayBetter (DB)}, and study it along with two variants: in \textsc{$\delta$-DelayBetter}, each delay must be of at most $\delta$; in ($\delta$-)\textsc{Path DB}, passengers also fully specify the vertices they should visit on their journey. 
On the positive side, we give a polynomial-time algorithm for \textsc{Path DB} and $\delta$-\textsc{Path DB}, and obtain as a corollary a polynomial-time algorithm for DB and $\delta$-DB on trees. We also provide an fpt algorithm for both problems parameterized by the size of the graph's Feedback Edge Set together with the number of passengers. 
On the negative side, we show NP-completeness of ($1$-)DB on bounded-degree temporal graphs even when the lifetime is $2$, and of ($10$-)DB on bounded-degree planar temporal graphs of lifetime $19$.
Our results complement previous work studying reachability problems in temporal graphs with delaying operations. 
This is to our knowledge the first such problem in which the aim is to facilitate travel between specific points (as opposed to facilitating or impeding a broadcast from one or many sources). 
\end{abstract}    

\section{Introduction}

In the first half of 2024, punctuality of Deutsche Bahn's long distance trains was 62.7\% \cite{deutschebahn_punctuality_2024}. Disruptions to train networks often result in passengers arriving later than planned or not at all. 
Whenever a train is late and the passengers of this train would miss a connecting train, there are two choices:
either, the second train departs on time, meaning that the passengers of the first train miss their connection, or the second train waits, meaning that the passengers can make the connection, at the cost of this train now also being late.
The problem of deciding whether (and by how much) such services should wait is the Delay Management problem, well studied in Operations Research. 

Separately, the field of temporal graph theory provides a general, rigorous mathematical framework with which to investigate the complexity due to the intrinsically dynamic properties of certain real-world networks. Briefly, a temporal graph is one whose edge set changes over time. Much work has been devoted to \emph{modification} problems of the form ``Given a temporal graph $\tgraph$, apply some (delaying or other) operations to satisfy some reachability property'' (see Table \ref{table:problems}), but interestingly the problem of managing delays to ensure that specific passengers arrive at their destination on time has yet to be studied in this framework. \cref{fig: intro} shows a simple example of a temporal graph illustrating such a scenario. 

The present work aims to study the practically interesting problem of Delay Management through the lens of temporal graph theory. We introduce the decision problem \textsc{DelayBetter} (or simply \textsc{DB}) which asks, given a temporal graph and a collection of passengers on its vertices, each with a desired destination and arrival time, whether it is feasible to delay some edges of the graph to satisfy each of the passengers. We also consider variants of the problem: \textsc{Path-DB}, where passengers must be routed along specific edges prescribed in the input, \textsc{$\delta$-DB}, where each edge can be delayed by at most some fixed $\delta$, and $\delta$-\textsc{Path-DB} combines both constraints. We present (parameterized) tractability and hardness results for these problems, including on structurally restricted graph classes. 


\begin{figure}
    \centering
    \includegraphics[width=\linewidth, page=12]{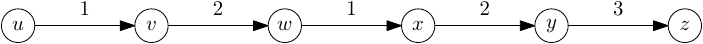}
    \caption{A temporal graph on 6 vertices. Consider the case where passengers at each of $u,v$, and $w$ wish to travel to each of $x,y$, and $z$ respectively, arriving at or before time 4. Then delaying the edge from $w$ to $x$ by at least $2$ is necessary for the two leftmost passengers to arrive on time, but entails that the passenger starting at $w$ cannot arrive at $z$ before time $5$. }
    \label{fig: intro}
\end{figure}



\def\arraystretch{1.5}
\begin{table}[!ht]
		\centering
		\setlength{\tabcolsep}{1pt}
		\begin{tabular}{
				| m{.23\linewidth}
				| >{\centering\arraybackslash}m{.14\linewidth} 
                | >{\centering\arraybackslash}m{.2\linewidth} 
				| >{\centering\arraybackslash}m{.21\linewidth} 
                    | >{\centering\arraybackslash}m{.2\linewidth} |}
                    
			\hline
            \textbf{Problem} 
            & \textbf{Operation} 
            & \textbf{Restriction}
            & \textbf{Reachability condition} 
            & \textbf{Additional inputs}
            \\
			\hhline{|=|=|=|=|=|}
			{\textsc{ReachFast} \cite{deligkas_minimizing_2023}}
            & Shift $(+-)$
            & N/A
            & $\forall x \in S: R_x=V$
        & sources $S \subseteq V$, $\tau \in \mathbb{N}$ to be minimized
            \\
            \hline
            {\textsc{TRLP} \cite{enright_reachability_2025}}            
            & Shift $(+-)$
            & up to $\eta$ edges, by up to $\delta$ each
            & $|R_x| \geq k$
            & designated source $x \in V$, $\eta, \delta, k \in \NN$
            \\
			\hline
            \textsc{MinReachDelay} \cite{molter_temporal_2024}
            & Delay
            & up to $\eta$ time-edges by exactly $\delta$ each
            & $|R_S| \leq k$
            & sources $S \subseteq V$, $\eta, \delta, k \in \NN$
            \\
			\hline
            \textsc{MinReach} \cite{deligkas_optimizing_2022}
            & Delay
            & up to $\eta$ time-edges by up to $\delta$ each
            & $|R_S| \leq k$
            & sources $S \subseteq V$, $\eta, \delta, k \in \NN$
            \\
			\hline
           {\textsc{MaxMinTaRDiS} }\cite{kutner_temporal_2024}
            & Choose time-labels
            & lifetime is $\tau$
            & $\nexists S \subseteq V, |S| < k: R_S = V$
            & $\tau,k \in \NN$
            \\
			\hline
            \textsc{($\delta$-)DelayBetter}  
            & Delay
            & (by up to $\delta$ per edge)
            & $(u,v,t) \in D \newline \Rightarrow v \in R_u^t$
            & $D \subseteq V \times V \times \NN$ 
            ($\delta \in \NN$)
            \\
            \hline
            \textsc{($\delta$-)Path DB}
            & Delay
            & (by up to $\delta$ per edge)
            & as above along specified path
            & $D \subseteq V \times V \times \NN \times 2^E$
            ($\delta \in \NN$)
            \\
            \hline
            
		\end{tabular}\medskip
		\caption{
        Comparison of our problems \textsc{DelayBetter} and \textsc{Path DelayBetter/Path DB} to problems in the literature. $R_u^t$ (resp.  $R_S^t$) denotes the set of vertices reachable from vertex $u$ (resp. any vertex $s \in S$) by time-step $t$ (when $t$ is the lifetime of the temporal graph, it is omitted).}\label{table:problems}
	\end{table}

\subsection{Problem setting}\label{sec: problem setting}

We denote $[i,j]$ the integer interval $\{i,i+1,\ldots,j\}$, and say a graph is \emph{cubic} if all vertices in the graph have degree 3. We now give some definitions from temporal graph theory. 

\begin{definition}[Temporal graph, temporal path]
    A \emph{temporal graph} $\tgraph=(G, \lambda)$ consists of a static graph $G$ (also called its \emph{footprint}, denoted $\footprint$) and a \emph{temporal assignment} $\lambda:E(G)\to \NN$. 
    The \emph{lifetime} $\tau\in \NN$ of a temporal graph $(G, \lambda)$ is the maximum time assigned to any edge by $\lambda$, and the \emph{time-edges} of a temporal graph $\mathcal{E}(\tgraph)$ are $\{(u,v,\lambda(u,v)) | (u,v)\in E(G)\}$.
    A \emph{temporal path} is a path in $\footprint$ whose edges have strictly increasing time-labels, and the \emph{arrival time} of a temporal path is the time-label of its last edge.
    
\end{definition}

We take this opportunity to note a more general definition of temporal graphs is often studied, where each edge may have multiple time labels (\emph{non-simple} temporal graphs). Another variant applies a notion of temporal reachability which allows for the traversal of consecutive edges at nondecreasing (rather than strictly increasing) times (\emph{non-strict} temporal paths). For a discussion of these different models (in the undirected setting only) we refer the interested reader to \cite{casteigts_simple_2022}. Hardness results from the simple setting generalize to the non-simple setting, and tractability for the non-simple setting may be applied to the simple setting. 
In the present work, we focus exclusively on strict temporal paths and simple (directed or undirected) temporal graphs.


\begin{definition}[Delaying]
    We say that a temporal assignment $\lambda'$ is a \emph{delaying} of an assignment $\lambda$ if $\lambda'(e)\ge \lambda(e)$ for every $e$.
    If $\lambda'(e)=\lambda(e)+\delta$, we say that $e$ is delayed by $\delta$ in $\lambda'$, and that $\lambda'$ is a $\delta$-delaying of $\lambda$ if every $e$ is delayed by at most $\delta$ in $\lambda'$ (hence a $\delta$-delaying is also a $(\delta+1)$-delaying). 
    
\end{definition}

We can now introduce our protagonists:

\begin{framed}
    \noindent\underline{\textsc{($\delta$-)DelayBetter}}\\ 
    {\bf Instance}: Temporal graph $\mathcal{G}=(G,\lambda)$, demands $D \subseteq V(G) \times V(G) \times \NN$ (, $\delta \in \NN$).\\
    {\bf Question}: Does there exist a ($\delta$-)delaying $\lambda'$ of $\lambda$ such that for each $(u,v,t) \in D$ there is a temporal path from $u$ to $v$ with arrival time at most $t$ in $(G,\lambda')$?
\end{framed}

\begin{framed}
    \noindent\underline{($\delta$-)\textsc{Path DelayBetter}}\\
    {\bf Instance}: Temporal graph $\mathcal{G}=(G,\lambda)$, demands $D \subseteq V(G) \times V(G) \times \NN \times 2^{E(G)}$ (,~$\delta \in \NN$).\\
    {\bf Question}: Does there exist a ($\delta$-)delaying $\lambda'$ of $\lambda$ such that for each $(u, v, t, P) \in D$ there is a temporal path from $u$ to $v$ in $(G,\lambda')$ with arrival time at most $t$ and footprint~$P$?
\end{framed}

We say a temporal graph $\tgraph$ is \emph{planar} if its footprint $\footprint$ is planar, and \emph{directed} (resp. \emph{undirected}) if $\footprint$ is directed (resp. undirected).
We use the shorthand \textsc{DB} for our problems, referring to, for example, \textsc{3-DB}, \textsc{Path DB}, or \textsc{DB}. For a demand $d$, we denote $d=(d_s,d_z,d_t)$. 
Restriction of, or parameterization by, the lifetime $\tau$ is often leveraged to obtain tractability of temporal graph problems. In our case, we denote by $\tinit$ the \emph{initial lifetime} (that of the temporal graph $\tgraph$ before delays are applied), and by $T_{\max}$ the latest arrival time required by any single demand -- that is, $\max_{d\in D}d_t$. We call $T_{\max}$ \emph{final lifetime} because it upper-bounds the lifetime of the temporal graph $(G,\lambda')$ (after delays are applied): any time-edge delayed beyond $T_{\max}$ in some feasible solution could instead be delayed to $T_{\max}$ instead (or not at all), since it will not be used by any passengers. For the same reason, we may assume without loss that $\tinit$ is at most $T_{\max}$.


\label{sec: related work}
\subsection{Related work}

\paragraph*{Temporal Graphs.}\label{sec: temporal graphs}
As we touched on earlier, modifying (or choosing) $\lambda$ to optimize a notion of reachability is a well-studied problem in temporal graph theory. Broadly, problems in this paradigm may either aim to \emph{worsen} or \emph{improve} the input temporal graph's connectedness. Problems in the first category (including  \textsc{MinReach} \cite{deligkas_optimizing_2022} and \textsc{MinReachDelay} \cite{molter_temporal_2024}) are typically motivated by practical cases where spread is undesired, such as epidemics. In the case of transportation networks, where connectedness is desired, the second category (which contains \textsc{TRLP} \cite{enright_reachability_2025} and \textsc{MaxMinTaRDiS} \cite{kutner_temporal_2024}) is of greater relevance. Of course, if the delays are controlled by an adversary, the opposite motivation becomes relevant to each problem: is there any strategy for the adversary to disconnect a transporation network, or facilitate disease spread? A related, but slightly different perspective on delays in temporal graphs is explored in \cite{fuchsle_temporal_2022} and  \cite{fuchsle_delay-robust_2022}, who determine how robust against unforeseen delays a given temporal graph is with and without re-routing of the passengers, respectively.

\paragraph*{Delay Management.}
\label{sec: delay management}
The Delay Management (DM) problem concerns itself with finding a good delaying strategy in a public transport network to minimize passenger inconvenience.
Usually, this means minimizing the total passenger delay, but other objectives like simultaneously minimizing the number of delayed trains or the operational costs have also been studied.
In the original problem, as introduced by \cite{schobel_model_2001}, passengers stick with their initial routes (as in \textsc{Path-DB}); a popular variant of the problem allows passenger re-routing (as in \textsc{DB}) \cite{dollevoet_delay_2012}. Both settings have since been the subject of much study, spanning both theory and practice. 

On the theoretical side, different models and algorithmic approaches have been introduced over the years \cite{binder_multi-objective_2017,ginkel_wait_2007,heilporn_optimization_2008,veelenturf_railway_2016,zhu_integrated_2020}. 
Due to modeling differences, studies of the computational complexity of different DM problem variants \cite{kanade_railway_2004,hutchison_computational_2005,schachtebeck_delay_2010} do not necessarily yield results for our problems.
In addition to minimizing an aggregate function (e.g., total weighted passenger delay \cite{kanade_railway_2004,hutchison_computational_2005}) rather than asking whether some specific set of passenger demands can be satisfied (as we do), DM problems are commonly formalized using \emph{event-activity networks} -- which are more expressive than temporal graphs. 
For example, the definition of DM in \cite{schachtebeck_delay_2010} includes \emph{headway constraints} (where two trains cannot use the same track segment simultaneously). Several interesting practically-motivated extensions are studied in this line of work, including a setting with slack times (trains may catch up on their delay), which makes the problem hard when the rail network is a line \cite{hutchison_computational_2005}, and the incorporation of rolling-stock circulation into the problem \cite{schachtebeck_delay_2010} -- though results in such settings do not straightforwardly translate into our model. 
Nonetheless, some results from these works can be adapted into the our setting; for example, Theorem 6.1 in \cite{kanade_railway_2004} could be adapted to show that \delaybetter ~is NP-complete in the directed setting with $T_{\max}=3$ (we strengthen this in \cref{thm:dir-np-c-lifetime2}).
Also, all of these works consider a directed model (as is natural for rail networks), whereas our results are proven for both directed and undirected temporal graphs.

On the more practical side, there have been a number of case-studies and data-driven approaches to this problem \cite{carosi_delay_2015,malucelli_delay_2019,zhang_train_2023}. 
In \cite{scozzaro_optimizing_2023}, a model for optimizing delays in rail and air travel combined is proposed, together with a European case study.
For a more comprehensive overview of the work in delay management, we refer the reader to \cite{cacchiani_overview_2014}, \cite{konig_review_2020}, and \cite{schobel_anita_optimization_2006}. A related area of research is the Timetabling Problem, which concerns itself with designing a timetable that is robust against delays. We refer the reader to \cite{geraets_cyclic_2007} for an introduction.

\label{sec: contribution}
\subsection{Our contribution}

We introduce the problems \textsc{($\delta$-)DelayBetter} and \textsc{($\delta$-)Path DB}, presenting (to our knowledge for the first time) a temporal graph-theoretic approach to the well-studied Delay Management problem.
On the positive side, we give a polynomial-time algorithm for ($\delta$-) \textsc{Path-DB}, and tractability for ($\delta$-)\textsc{DelayBetter} on trees as a corollary. Later, we leverage this algorithm to obtain a fixed-parameter tractable (fpt) algorithm parameterized by the number of demands and the size of the feedback edge set of the footprint graph. On the negative side, we establish that \textsc{DelayBetter} remains NP-complete on inputs with $T_{\max}=2$ in both the directed and undirected setting (which entails that $1$-\delaybetter ~is NP-complete under the same constraint).
Moreover, we show that the problem remains hard on planar (directed or undirected) temporal graphs with $\tmax=19$, even when $\delta=10$. Our results provide a first insight into the structural restrictions which do (and do not) suffice to guarantee tractability of this natural problem. 
Proofs of statements marked $(\ast)$ can be found in the appendix at the end of the paper.

\subsection{Discussion and open questions}

We show that ($\delta$-)\textsc{Path DB} is in P and that ($\delta$-)\delaybetter ~is fpt parameterized by $|D|+\fes$, where $\fes$ is the size of smallest feedback edge set of (the undirected version of) $\footprint$. 
It seems likely that the techniques used in those proofs could actually solve a broader family of problems -- including, for example, the natural extension of \delaybetter ~wherein demands specify a departure time as well as an arrival time, but also possibly problems which do not specify individual demands as part of the input. 
Can dependence on $|D|$ be eliminated from our fpt result? If not, then what structural parameter is sufficient to yield an fpt result without requiring $|D|$ as a parameter?
A more general question for future study is: what family of temporal graph modification problems admit an fpt algorithm in the size of the feedback edge set? 
Separately, what is complexity of the problems parameterized by fine-grained temporal parameters (e.g., vertex interval membership width \cite{bumpus_vimw_2021}), or by smaller structural parameters than $\rho$ (e.g., the feedback vertex number)? 
We note that for directed graphs, the size of a minimum feedback arc set (the deletion of which leaves a directed acyclic graph, or DAG) is insufficient, since we show in \cref{thm:planar} that the problem is NP-complete restricted to (planar) DAGs. 

Another question we leave open is: what is the complexity of \delaybetter ~restricted to planar inputs with $\tmax \in [2,18]$? 
(Our proofs of Theorems  \ref{thm:undir-np-c-lifetime2} and \ref{thm:dir-np-c-lifetime2} do not preserve planarity, and moreover reduce from a variant of \textsc{NAE 3SAT}, the restriction of which to planar instances is solvable in polynomial time \cite{moret_planarNAE3SAT_1988}.)
Also stemming from our planar proof is the question of whether $\delta$-\delaybetter ~restricted to planar graphs is computationally easy or hard for values of $\delta$ below $10$. Our proof was aimed at minimizing $\tmax$ while retaining planarity, so we expect that some easy adjustments to it might yield hardness for, e.g., $\delta=9$, but we expect different techniques are necessary to deal with the case of $\delta=1$ on planar graphs.

Yet another direction our investigation could be extended is to consider \emph{non-simple} temporal graphs. Our hardness results extend immediately to this case, but our algorithms do not -- in the non-simple setting, we do not expect our linear programming approach to work, and it is not even obvious whether our problems would be tractable restricted to trees.

Lastly, we observe that our results for directed and undirected versions of the problem are the same. This is particularly surprising because some of our results require substantially different proofs for each setting. An open question for future work is then: are there any natural restrictions on the input which entail instances are tractable in the directed case and computationally hard in the undirected case (or vice versa)?

\section{Preliminary Results}

We begin with some basic results, the proofs of which may help to familiarize the reader with the behavior of our problems.
We first establish a useful relation between \textsc{$\delta$-DB} and \textsc{DelayBetter}.
Clearly, \delaybetter ~is reducible to $\delta$-\delaybetter, by simply assigning  a sufficiently large value to $\delta$ (e.g., the final lifetime $\tmax$ of the \delaybetter ~instance). Interestingly, the converse also holds:

\begin{lemma}\label{lem:equiv}
    For any $\delta\in \mathbb N$, $\delta$-\textsc{DelayBetter} is reducible in linear time to \textsc{DelayBetter}. If the input instance is planar (resp. has bounded final lifetime) then the same holds for the output.
\end{lemma}

\begin{proof}
    We require different reductions for directed and undirected graphs. In both cases, substitute a gadget in place of each edge in the original instance, and increase the lifetime of the instance. Both constructions preserve planarity, and the \delaybetter-instance has final lifetime at most $2 T_{\max} + 2 \delta + 1$ (directed) or $T_{\max} + \delta + 2$ (undirected), where $T_{\max}$ is the final lifetime of the $\delta$-\textsc{DelayBetter} instance.

    We first deal with the undirected case. We begin with a $\delta$-\delaybetter ~instance  $((G,\lambda), D, \delta)$ having the property that every time is at time $3$ or later (if necessary, this can be achieved by uniformly incrementing all times in the demands and in the temporal assignment by $2$).  
    We then create, for every time-edge $(u,v,t)$, a gadget on $3\delta+3$ new vertices $\{uv_t, \ldots, uv_{t+\delta}, u_1, \ldots, u_\delta, v_1, \ldots, v_\delta, u', v'\}$ and $\delta+1$ new demands  $\{(uv_i, v',i+1) : i \in [t, t+\delta]\}$, as shown in \cref{fig:undir-equiv}.
    
    \begin{figure}[!ht]
        \centering
        \includegraphics[page=20, width=\linewidth]{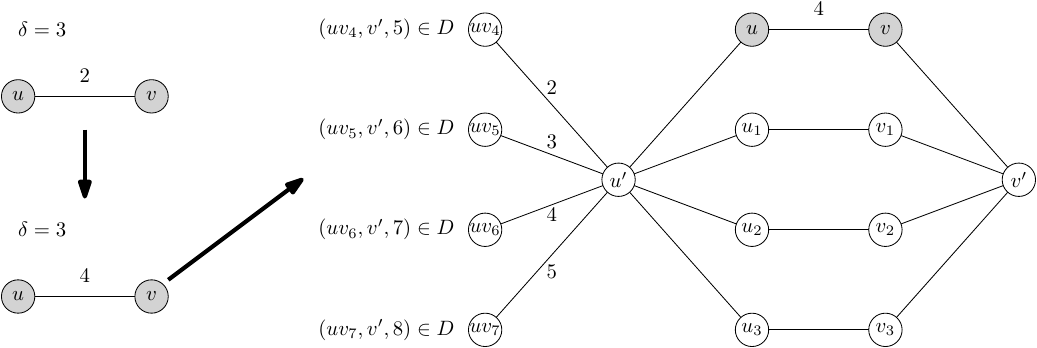}
        \caption{A sketch of our reduction for undirected temporal graphs for a time-edge $(u,v,2)$ in an instance with $\delta=3$. For readability, edges assigned time $1$ in the output instance are unlabeled.}
        \label{fig:undir-equiv}
    \end{figure}
    
    Now each of the $\delta+1$ demands into $v'$ must be routed through a different path. Because there are $\delta+1$ possible paths in total, some demand $(uv_i,v',t)$ must be routed through the edge $(u,v)$ -- and this entails that $\lambda'(u,v)=i$, yielding the desired result since $i\in[t, t+\delta]$ by construction.

    We now give the proof for the directed case.
    Given an instance $(\tgraph=(G,\lambda), D, \delta)$ of $\delta$-\delaybetter, we produce an instance $(\tgraph'=(G',\lambda'), D')$ of \delaybetter ~as follows. (For this proof, we use $\lambda'$ to refer to the initial assignment of the new instance, not the delaying of $\lambda$.)
    We first include $\{(u,v,2t) | (u,v,t) \in D\}$ as demands, which we call \emph{travelers}. 
    Next, we replace every time-edge $(u,v)$ at time $t$ with the gadget pictured in Figure \ref{fig:dir-equiv}, and add the demand $(u',v',2t+2\delta+1)$. We call demands introduced in this step \emph{hermits}, the edge $(u',u)$ that hermit's \emph{trailhead}, and the edge $(u,uv)$ (resp. $(uv,v)$) a \emph{first-half} (resp. \emph{second-half}) edge. This concludes the construction. 

    \begin{figure}
        \centering
        \includegraphics[page=2, width=\linewidth]{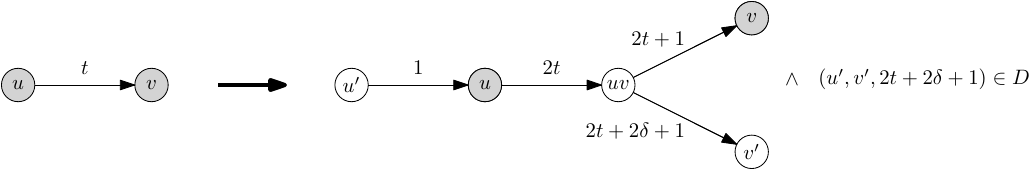}
        \caption{A sketch of our reduction for directed temporal graphs.}
        \label{fig:dir-equiv}
    \end{figure}

    Clearly, if the \textsc{$\delta$-DB} instance reduced from was a yes-instance, then the \textsc{DelayBetter} instance obtained is also a yes-instance: whenever some edge $(u,v)$ is delayed by some amount $x$ in the original instance, delay both $(u,uv)$ and $(uv,v)$ by $2x$. It remains to show the converse.

    Let $\lambda^*$ be a solution to our modified problem with a pareto-optimal time-assignment of the edges (that is, one such that there is no other solution whose time-labels are all strictly smaller or equal to those under $\lambda^*$). 

    \begin{restatable}[$\ast$]{claim}{earlyHermitsClaim}\label{clm:hermits-leave-early}
        Hermits leave early: if $e$ is a hermit's trailhead, then $\lambda^*(e)=1$.
    \end{restatable}


    \begin{claim}\label{clm:even-and-odd}
        Under $\lambda^*$, every first-half edge (resp. second-half edge) is assigned an even (resp. odd) time.
    \end{claim}

    \begin{claimproof}
        Suppose otherwise. We must deal with two cases:
        
        We deal first with the case where the earliest edge violating this claim is a first-half at an odd time. Let $(u,uv)$ be the earliest first-half edge assigned an odd time (say, $t'$) under $\lambda^*$. By pareto-optimality of $\lambda^*$ is must be that assigning time $t'-1$ to the edge $(u,uv)$ would stop some demand from being satisfied. This demand cannot be a hermit because of \cref{clm:hermits-leave-early} -- so there must be some traveler arriving at $u$ at time $t'-1$, a contradiction since our premise for this case was that $(u,uv)$ was the earliest edge violating the claim.
    
        Suppose instead that the earliest offending edge is a second-half edge $(wu,u)$ assigned an even time (say, $t'-1$). We again quickly find that this can only be due to the first-half edge $(w,wu)$ being used by a traveler at time $t'-2$ - again reaching a contradiction, since this is a strictly earlier odd time assigned to a first-half edge.
    \end{claimproof}
    
    \begin{restatable}[$\ast$]{claim}{maxDelayTwoDeltaClaim}\label{clm:maxdelay-2delta}
        For each time-edge $(u,v,t)$ in the original instance, $\lambda^*(u,uv)\in [2t,2t+2\delta]$.
    \end{restatable}
    
    \cref{clm:even-and-odd} allows us to recover a time-labeling for our initial $\delta$-\textsc{DelayBetter} instance by assigning to each the edge $(u,v)$  the time $\frac{(u,uv)}{2}$ while preserving the temporal paths of travelers. \cref{clm:maxdelay-2delta} entails that this time-labeling does not delay any edge by more than $\delta$, and the result follows.
\end{proof}

\begin{restatable}[$\ast$]{lemma}{polyTmaxLemma}\label{lem:poly-tmax}
    An instance of \textsc{DelayBetter} or \textsc{Path DB} (resp. $\delta$-\textsc{DelayBetter} or $\delta$-\textsc{Path DB}) may be reduced in polynomial time to an instance of the same problem with $T_{\max} \in \mathrm{poly}(n)$ (resp. $T_{\max} \in\mathrm{poly}(n+\delta)$). 
\end{restatable}

\begin{restatable}[$\ast$]{lemma}{inNPLemma}\label{lem:in-np}
    \textsc{($\delta$-)DelayBetter} is contained in NP.
\end{restatable}

\section{Tractability Results}

We begin with a small positive result which can be obtained easily from prior work. 
\begin{lemma}\label{lem:single-source-poly}
    \textsc{DelayBetter} is solvable in polynomial time when $\lambda$ is the constant function $\bm 1$ and all demands in $D$ have the same source.
\end{lemma}

\begin{proof}
    We use the One Source Reach Fast algorithm from \cite{deligkas_minimizing_2023}: They show that the time-assignment of their algorithm computes, for a given source $v \in V$ and every remaining vertex $u \in V$, the individual minimum time that $v$ needs to reach $u$. 
    If this computed minimum time is at most our demanded arrival time for all demands $(v,u,t) \in D$, then we have a YES-instance, otherwise we have a NO-instance.
 \end{proof}

We now turn to the case where passenger demands fully prescribe the path they must be routed along, establishing tractability through a linear programming argument.

\begin{theorem} \label{thm: path db is in p}
\textsc{Path DelayBetter} and $\delta$-\textsc{Path DelayBetter} are both in P.
\end{theorem}

\begin{proof}
    Let $((G, \lambda), D)$ be an instance of \textsc{Path DelayBetter} (or, let $((G, \lambda), D, \delta)$ be an instance of $\delta$-\textsc{Path-DelayBetter} -- the proof differs only in a few details). 

    We begin by introducing some notation. Our proof is for directed and undirected inputs -- we shall use $uv$ to mean the edge $(u,v)$, but in the undirected case $uv=vu$ whereas in the directed case these are $uv \ne vu$.
    For a demand $d \in D$, we denote by $d_P$ the specified static path in $G$ from $d_s$ to $d_z$, and $d_f$ the final edge of $d_P$, which is incident to $d_z$ and must be at time $d_t$ or earlier to satisfy the demand.
    We also use $t_{uv}$ as shorthand for $\lambda(u,v)$, and $t'_{uv}$ for $\lambda'(u,v)$. Lastly, we 
    define the relation $(u,v) \prec (v,w)$, to be true if and only if $(u,v)$ immediately precedes $(v,w)$ in the path $d_P$ for some $d \in D$.

    Consider the following linear program:
    \begin{align}
    &\mathrm{maximize} \sum_{d\in D} d_t - t_{d_f}' \text{ , subject to}\\
	&t_{uv} \le t'_{uv} \text{ for each } (u,v)\in E(G) \label{eq:delay}\\
	&t'_{uv} \le  t_{uv} + \delta \text{ for each } (u,v)\in E(G) \text{ (only for $\delta$-\textsc{Path-DB})} \label{eq:delta-delay}\\
	&t_{uv}' \le t'_{vw}+1 \text{ for each pair of edges $uv$ and $vw$ such that $uv \prec vw$} \label{eq:strict-path}\\
	&t_{d_f}' \le d_t \label{eq:arrival}\text{ for each demand $d$}
    \end{align}
    
    This LP has $\{t'_{uv}: (u,v) \in E(G)\}$ as its set of \emph{unknown variables}. The variables $\{t_{uv} : (u,v) \in E(G)\} \cup \{d_t : d \in D\}$ correspond to given integers fully specified by the \textsc{Path-DB} instance $((G,\lambda), D)$ (and $d_f$ likewise refers to a specific edge of $G$).

    \begin{restatable}[$\ast$]{claim}{integralLPClaim}\label{clm:integral-lp}
        The LP is integral. Meaning: at least one optimal solution of the LP assigns integers to all of its unknown variables. Moreover, an integral solution may be recovered from a non-integral solution in polynomial time.
    \end{restatable}
    
    Since linear programs are solvable in polynomial time \cite{karmarkay1984linearprog}, we may first solve the LP and then (if the solution is not already integral) apply Claim \ref{clm:integral-lp} to recover an integral solution.
    We note here that a modification of Kahn's algorithm \cite{kahn_topological_1962} for topological sorting may be used to compute a solution to this particular LP directly and more efficiently. A detailed proof would be quite technical and incongruous with the rest of the paper, so has been omitted.

    An integral solution to this LP fully specifies a delaying $\lambda'$ satisfying the ($\delta$-)\textsc{Path DB} instance.
    Note that: $\lambda'$ is indeed a ($\delta$-)delaying of $\lambda$ (due to \cref{eq:delay,eq:delta-delay}); enables strict temporal paths along each path specified in $D$ (due to \cref{eq:strict-path}); and that each of these paths reaches the destination vertex by the arrival time prescribed (due to \cref{eq:arrival}). Conversely, it should be clear that any delaying $\lambda'$ satisfying the ($\delta$-)\textsc{Path DB} instance specifies a (not necessarily optimal) solution to the LP. In fact, the LP allows us not only to \emph{decide} ($\delta$-)\textsc{Path DB}, but more strongly to solve its optimization variant.
\end{proof}

Since trees are characterized by any pair $(u,v)$ being connected by a unique (static) path, we obtain the following corollary:

\begin{corollary} \label{cor: db on trees is in p}
    ($\delta$-)\textsc{DelayBetter} is in P when the underlying graph $\footprint$ is a (directed)~tree.
\end{corollary}

Next, we are able to extend this result to ``tree-like'' graphs, by parameterizing by the size of the instance's feedback edge set.

\begin{theorem}\label{thm:fpt-fes}
    On directed (reps. undirected) temporal graphs, with $|FES(\mathcal{G}_\downarrow)|=\fes$, ($\delta$-)\textsc{DelayBetter} is solvable in time $O(\fes! \cdot 2^{\fes\cdot|D|} \cdot \mathrm{poly}(n))$ (resp. $O(\fes! \cdot 3^{\fes\cdot|D|}\cdot\mathrm{poly}(n))$).
\end{theorem}
\begin{proof}
    The proofs for directed and undirected graphs differ only in small details, and those for for $\delta$-DB and \textsc{DelayBetter} are identical (until we apply \cref{thm: path db is in p}).
    
    Let $E'$ be a feedback edge set of (the undirected version of) $\footprint$ of size $\fes$. We iterate over each of the $\fes!$ possible orderings $(e_1, e_2, ... e_{\fes})$ of $E'$, and require that $t_{e_1} \leq t_{e_2} \leq ... \leq t_{e_{\fes}}$. (Note that if $T_{\max}$ is small and $\fes$ is large, we may prefer to iterate over all $(T_{\max})^\fes$ assignments and obtain an ordering from those.)

    In any solution, each demand $d\in D$ is satisfied by a strict temporal path from $d_u$ to $d_v$ using some subset of the edges of $E'$. In the directed case, specifying this subset (together with the ordering fixed earlier) fully specifies the path from $d_u$ to $d_v$; in the undirected case, it is also necessary to specify the direction taken for each edge. The journey from one edge in the subset to the next is uniquely determined due to the fact that it can only use the edges of the spanning tree obtained by removing $E'$ from $\tgraph$.

    For directed graphs, this means there are at most $2^{\fes}$ possible paths for each demand (an edge is either chosen or not), and thus $2^{\fes \cdot |D|}$ for all demands.
    For undirected graphs, we get $3^{\fes}$ possible paths per demand (an edge $(u,v)\in E'$ is either traversed from $u$ to $v$, from $v$ to $u$, or not at all), and thus $3^{\fes \cdot |D|}$ for all demands. 
    For each ordering of $E'$ and collection of subsets of $E'$, there is a corresponding instance of ($\delta$-)\textsc{Path DB}.
    
    In total, it is sufficient to solve $\fes! \cdot 2^{\fes \cdot |D|}$ such instances of ($\delta$-)\textsc{Path DB} for directed graphs, and $\fes! \cdot 3^{\fes \cdot |D|}$ instances of ($\delta$-)\textsc{Path DB} for undirected graphs.
    Since ($\delta$-)\textsc{Path DB} is solvable in polynomial time by Theorem \ref{thm: path db is in p}, we obtain the desired result.
 \end{proof}

\section{Hardness results}

Our first two hardness results are in the restrictive setting wherein $T_{\max}=2$ and the initial temporal assignment is the constant function $\bf 1$. 
In this setting, the problems \delaybetter ~and $\delta$-\delaybetter ~essentially ask only whether there \emph{exists} any $\lambda$ satisfying our passenger demands; any such $\lambda$ can be assumed \wlogsm to have lifetime $2$, and could be obtained by delaying all time-edges by at most $1$ - meaning our results hold for any $\delta \ge 1$.

\begin{theorem}\label{thm:undir-np-c-lifetime2}
    On undirected graphs, \delaybetter ~(and $\delta$-\delaybetter ~with any $\delta\ge 1$) is NP-complete even restricted to instances where $T_{\max}=2$, the initial temporal assignment is the constant function $\bf 1$, and the $\footprint$ has diameter 6.
\end{theorem}

\begin{proof}
    Our reduction is from \textsc{Positive Not-All-Equal Exactly 3SAT} \cite{antunes2019characterizing}, an NP-complete problem taking as input a formula $\phi$ consisting of triples of variables (which appear only positively). The formula $\phi$ is a yes-instance if there is an assignment to the variables such that every triple contains at least one true variable and at least one false variable.  

    We shall construct a graph $G$ which admits a temporal assignment $\lambda': E(G) \to \{1,2\}$ satisfying all our demands if and only if $\phi$ admits a satisfying assignment. 
    Figure \ref{fig:lifetime-2-undirected} may be of use to the reader in following the proof. Solid (resp. dashed) edges in bold are ones which are necessarily assigned $1$ (resp. $2$) in any temporal assignment $\lambda'$ satisfying all demands.

    \begin{figure}[!ht]
        \centering
        \includegraphics[width=0.8\linewidth,page=17]{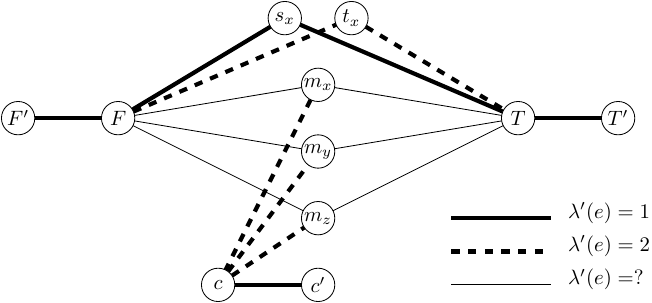}
        \caption{A sketch of our construction. The vertices $s_x, t_x, m_x$ constitute the gadget for a variable $x$, and the vertices $m_x,m_y,m_z,c,c'$ constitute the gadget for a triple $c=\mathrm{nae}(x,y,z)$.}
        \label{fig:lifetime-2-undirected}
    \end{figure}

    We shall refer to the demand $(u,v,t)$ as a $t$-demand from $u$ to $v$.     
    We begin with four special vertices $F,F',T,T'$, with $1$-demands from $F'$ to $F$ and $T'$ to $T$.   
    Then, for each variable $x$ in $\phi$, we introduce vertices $s_x,t_x,m_x$ and edges from each of these to each of $T,F$. We further introduce 1-demands from $s_x$ to each of $T,F$ (enforcing that both edges must be assigned time $1$) and 2-demands from each of $T',F'$ to $t_x$ (enforcing that both of $(T,t_x),(F,t_x)$ must be assigned time $2$). Lastly, we introduce 2-demands from $s_x$ to $m_x$ and from $m_x$ to $t_x$, which together with the previous constraints, guarantees that $\lambda'(m_x,T)\ne\lambda'(m_x,F)$. 
    
    Next, for each triple $c$ in $\phi$, we create vertices $c$ and $c'$ and a 1-demand between these, and connect the vertex $c$ to $m_x$ by an edge if $x$ appears in the triple $c$ and introduce a 2-demand from $c'$ to $m_x$. We also introduce $2$-demands from each of $T$ and $F$ to $c$.
    
    The intention is that assigning $\lambda'(m_x,T)=1$ will correspond to an assignment of \texttt{true} to $x$ in $\phi$, and assigning $\lambda'(m_x,F)=1$ will correspond to an assignment of \texttt{false} to $x$ in $\phi$. Suppose that some $\lambda'$ satisfies all demands. Then the assignment in which variable $x$ is set to \texttt{true} if $\lambda'(m_x,T)=1$ and \texttt{false} otherwise is a satisfying assignment of $\phi$.

    Suppose that $\phi$ has a satisfying assignment $X$. Consider the temporal assignment $\lambda'$ in which $\lambda'(m_x,T)=1$ and $\lambda'(m_x,F)=2$ if $x$ is \texttt{true} under $X$ and $\lambda'(m_x,T)=2$ and $\lambda'(m_x,F)=1$ otherwise (and all other values of $\lambda'$ are as specified in Figure \ref{fig:lifetime-2-undirected}). 
    Under $\lambda'$, every clause $c$ is adjacent to some pair of vertices $m_x,m_y$ such that $x=\texttt{true}$ under $X$ and  $y=\texttt{false}$ under $X$ -- so the 2-demand from $T$ (resp. $F$) to $c$ can be routed through $m_x$ (resp. $m_y$). It is clear that $\lambda'$ satisfies all other demands. 
\end{proof}

Our result for directed graphs requires a slightly different proof:

\begin{theorem}\label{thm:dir-np-c-lifetime2}
    On directed graphs, \delaybetter ~(and $\delta$-\delaybetter ~with any $\delta\ge 1$) is NP-complete even restricted to instances where $T_{\max}=2$ and $G$ has no directed cycles.
\end{theorem}

\begin{proof}
    The reduction is again from \textsc{Positive Not-All-Equal Exactly 3SAT}. Given a formula $\phi$, we construct a directed graph $G$ as follows:
    Then, for each variable $x$ in $\phi$, we introduce six vertices $s_x, s_x^T, s_x^F, t_x, t_x^T, t_x^F$ and connect them as shown in Figure \ref{fig:lifetime-2-directed}. Further, we introduce a vertex $c$ identified with each triple $c$ in $\phi$, and create directed edges from $c$ to $s_x^T$ and from $c$ to $s_x^F$. 
    
    \begin{figure}[!ht]
        \centering
        \includegraphics[width=0.8\linewidth,page=19]{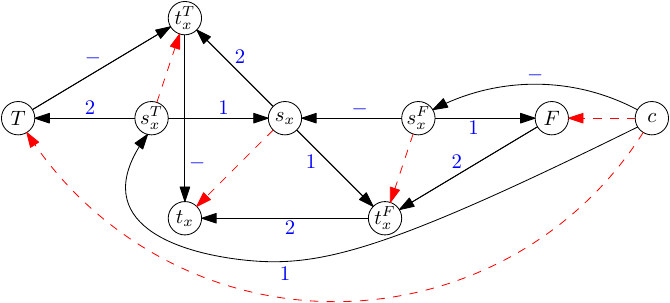}
        \caption{A sketch of our construction showing NP-completeness of \delaybetter ~for digraphs. The vertices $s_x, s_x^T, s_x^F, t_x, t_x^T, t_x^F$ constitute the gadget for a variable $x$, and the vertex $c$ (together with its out-edges) constitutes the gadget for a triple $c\ni x$. Directed edges in $G$ are solid, whereas 2-demands are shown as dashed arrows in red. The temporal assignment shown (in blue) is one corresponding to the assignment $x$=\texttt{true} ($-$ denotes an arbitrary choice).}
        \label{fig:lifetime-2-directed}
    \end{figure}

    We now specify the demands for our instance; for each variable $x$, we have 2-demands from $s_x$ (resp. $s_x^T$, $s_x^F$) to $t_x$ (resp. $t_x^T, s_x^F$), and for each clause $c$ we have 2-demands from $c$ to each of $T$ and $F$. (All of our demands are 2-demands, and these are shown as red dashed arrows in Figure \ref{fig:lifetime-2-directed}.) We let the constant function $\bf 1$ be the initial temporal assignment for our directed graph, and this concludes the construction of our ($\delta$-)\delaybetter ~instance (together with specifying $\delta=1$, if necessary).
    
    \begin{claim}\label{clm:lifetime-2-dir-nae}
        Let $\lambda'$ be any temporal assignment satisfying all demands in our construction. Then $\lambda'(s_x^T,T)=2$ entails $\lambda'(s_x^F,F)=1$, and $\lambda'(s_x^F,F)=2$ entails $\lambda'(s_x^T,T)=1$.
    \end{claim}

    \begin{claimproof}
        Suppose $\lambda'(s_x^T,T)=2$ for some $x$. Since all our demands are satisfied, we have that there must be a temporal path from $s_x^T$ to $t_x^T$ arriving at time 2. Such a path necessarily leaves at time 1 (since the two vertices are at distance 2). Consequently, $\lambda'(s_x^T,s_x)=1$ and $\lambda'(s_x, t_x^T)=2$. Similarly, we now must have that the 2-demand from $s_x$ to $t_x$ is routed through $t_x^F$, entailing that $\lambda'(s_x, t_x^F)=1$ and $\lambda'(t_x^F,t_x)=2$. Applying the same logic a third time, the 2-demand from $s_x^F$ to $t_x^F$ must be routed through $F$, and the desired claim follows. (The other direction is symmetric.)
    \end{claimproof}

    Suppose that some $\lambda'$ satisfies all demands. Consider the truth assignment in which a variable $x$ is set to \texttt{true} if $\lambda'(s_x^T,T)=2$, and \texttt{false} otherwise. Suppose for contradiction that under this truth assignment, some triple $c$ is not satisfied. Then either: (a) all variables in $c$ are \texttt{true} under our truth assignment, and leveraging Claim \ref{clm:lifetime-2-dir-nae}, the vertex $c$ cannot reach the vertex $F$ by time $2$; or, (b), all variables in $c$ are \texttt{false} under our truth assignment, and there $c$ cannot reach the vertex $T$ by time $2$. In either case, some demand is not satisfied and we derive the desired contradiction. 
    
    Now suppose that there is some truth assignment satisfying $\phi$. Consider the temporal assignment $\lambda'$ in which:
    \begin{itemize}
        \item If $x\in c$ and $x$ is \texttt{true} (resp. \texttt{false}) under the assignment, then $\lambda'(c, s_x^T)=1$ (resp. $\lambda'(c, s_x^F)=1$), and
        \item If $x$ is \texttt{true} (resp. \texttt{false}) under the truth assignment, then $\lambda'(s_x^T,T)=2$ (resp. $\lambda'(s_x^F,F)=2$) and temporal assignments to other directed edges in each variable gadget being chosen consistently with the proof of Claim \ref{clm:lifetime-2-dir-nae} to satisfy demands within the variable gadget, as shown in Figure \ref{fig:lifetime-2-directed}.
        \item All other edges are assigned times arbitrarily. 
    \end{itemize}

    Under $\lambda'$, $c$ has a path to $T$ (resp. $F$) through $s_x^T$ (resp. $s_x^F$) if and only if $x\in c$ is assigned \texttt{true} (resp. \texttt{false}). 
    It should be clear that $\lambda'$ satisfies all other demands in our instance by construction, and the result follows.
\end{proof}

Having shown that the instance being a tree yields  tractability in \cref{cor: db on trees is in p}, we consider the case of planar graphs - a well-studied superclass of trees. 

\begin{theorem}\label{thm:planar}
    $\delta$-\delaybetter ~is NP-complete under any combination of the following:
    \begin{itemize}
        \item $G$ is planar and has maximum degree $10$.
        \item Either $G$ is undirected, or $G$ is a directed acyclic graph (DAG).
        \item Either $T_{\max}=19$ and $\tinit=1$ (with any $\delta\ge 19$), or $T_{\max}=19$ and $\delta=10$.
    \end{itemize}
\end{theorem}

\begin{proof}
    Our reduction is from \marxcoloring ~(\marxcol). That problem asks, given an undirected graph $G$ (which is planar, bipartite, and cubic) and a precoloring of its edges $P:E(G) \to \{R,G,B,U\}$ (indicating red, green, blue, and uncolored edges respectively) whether there is a proper edge-coloring $C:E(G) \to \{R,G,B\}$ of $G$ such that $P(e)\in \{R,G,B\} \implies C(e) = P(e)$. 
    Let $A,B$ be an arbitrary bipartition of $V(G)$, and fix an arbitrary order on $V(G)$ (so we may refer to the $i$th neighbor of some vertex). 

    We shall make use of the following hardness result:

    \begin{lemma}[Theorem 2.3 in \cite{marx_np-completeness_2005}]
        \marxcoloring ~is NP-complete.
    \end{lemma}

    \paragraph*{Construction} 
    Our construction for the directed case is a specific orientation of our construction for the undirected case. Consequently, we shall describe the directed construction, which implicitly also specifies the undirected construction -- but still detail explicitly, for example, that edge-gadgets can only be traversed from an $A$-gadget to a $B$-gadget (which is trivial in the directed case). 

    In our construction, the inclusion of a \emph{bold time-edge} $(x,y,t)$ essentially dictates that the edge $(x,y)$ is assigned time $t$ exactly in any temporal assignment satisfying all demands. To realize this constraint, we introduce a temporal path of length and duration $t-1$ on new vertices $xy_1,\ldots,xy_{t-1}$ and $x$, as shown in \cref{fig:bold-time-edges} and include $(xy_1, y,t)$ in our demands. Note that in the case where $t=1$ no new vertices are created -- only the demand $(x,y,1)$. 

    \begin{figure}[!ht]
        \centering
        \includegraphics[page=16,width=\linewidth]{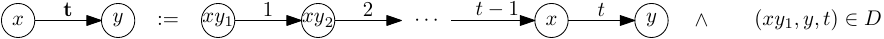}
        \caption{Our gadget ensuring that bold time-edges are never delayed.}
        \label{fig:bold-time-edges}
    \end{figure}

    The reader may find the diagram in \cref{fig:planar-reduction} helpful. 
    We first describe the graph $G'$ for our instance of \delaybetter, and then the demands $D$. (For now, we let the initial temporal assignment $\lambda$ be $1$ everywhere except for bold time-edges and their gadgets.) 

    \begin{figure}[!ht]
        \centering
        \includegraphics[page=21,width=1\linewidth]{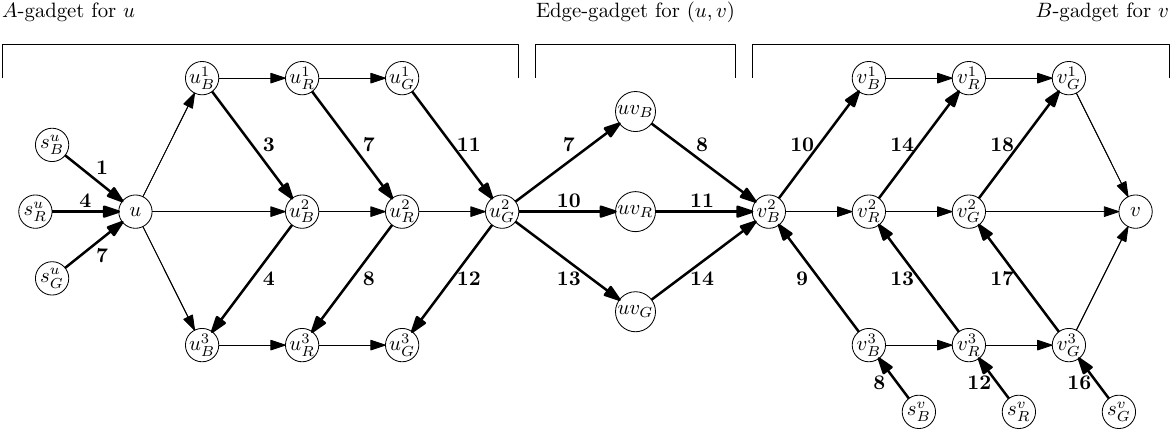}
        \caption{A sketch of our reduction from \marxcoloring ~to \delaybetter. Only bold time-edges are labeled. 
        }
        \label{fig:planar-reduction}
    \end{figure}

    For each vertex $v \in V(G)$, we create a \emph{vertex-gadget} consisting of a copy of $v$ and 12 other vertices $s_B^v,s_R^v,s_G^v,v_B^1,v_B^2,v_B^3,v_R^1,v_R^2,u_R^3,v_G^1,v_G^2,v_G^3$ (subscripts represent color; superscript $i$ represents the $i$th neighbor of $v$). These vertices are connected differently depending on whether $v \in A$ or $v \in B$, as shown in \cref{fig:planar-reduction}. In a vertex-gadget, we call \emph{spoke edges} those edges which are not bold, and \emph{blue} (resp. \emph{red, green}) \emph{layer} the vertices $v_B^i$ (resp. $v_R^i,v_G^i$).

    For each edge $(u,v) \in E(G)$ with $u\in A, v \in B$, we also create three vertices $uv_B, uv_R, uv_G$. 
    Then if $u$ is the $i$th neighbor of $v$ and $v$ is the $j$th neighbor of $u$, we introduce six bold time-edges $(u^i_G,uv_B,7),(uv_B,v_B^j,8),(u^i_G,uv_R,10),(uv_R,v_B^j,11),(u^i_G,uv_G,13),(uv_G,v^j_B,14)$. (In \cref{fig:planar-reduction} $u$ is the second neighbor of $v$ and vice versa.) If $P(u,v)$ is precolored $G$ under $P$, then we delete two of $uv_B,uv_R$, or $uv_G$, as appropriate, leaving just one path from $u$ to $v$.
    In \cref{fig:planar-reduction}: $u$ is the second neighbor of the $v$; $v$ is the second neighbor of $u$; and $P(u,v)=U$.
    

    We make use of three types of demands:
    \begin{description}
        \item[Bold demands] as described earlier and shown in \cref{fig:bold-time-edges}.
        \item[Hermits] demands from a vertex in a vertex-gadget to another vertex in the same gadget. For each vertex $u\in A$ we have demands $(s_B^u,u_B^3,4), (s_R^u,u_R^3,8),$ and $(s_G^u,u_G^3,12)$, and for each vertex $v \in B$ we have demands $(s_B^v, v, 13),(s_R^v, v, 16),$ and $(s_G^v, v, 19)$. (We say hermits have the color of the layer their source or destination lies in.)
        \item[Travelers] demands from a vertex in an $A$-gadget to a vertex in a $B$-gadget. For each edge $(u,v)$ in the  \marxcol ~instance with $u \in A$ and $v \in B$, we add a demand $(u,v,19)$.
    \end{description}

    This concludes our construction.

    \paragraph*{Correctness}

    \begin{claim}
        If the \marxcol ~instance $G,P$ is a yes-instance, then the \delaybetter ~instance $(G',\lambda),D$ is a yes-instance.
    \end{claim}

    \begin{claimproof}
        We shall construct a delaying $\lambda'$ of the initial temporal assignment $\lambda$ satisfying all demands in $D$. 
        Consider a proper edge coloring $C$ of $G$ which extends $P$. 

        First, we do not delay any bold time-edges -- i.e., for those, $\lambda(e)=\lambda'(e)$. Note that all bold demands are immediately satisfied under any such labeling. 
        
        Let $(u,v)$ be an edge assigned color $B$ (resp. $R,G$) under $C$, with $u$ being the $i$th neighbor of $v$ and $v$ being the $j$th neighbor of $u$. We assign:
        \begin{itemize}
            \item $\lambda'(u,u^j_B)=2$ (resp. $5,8$)
            \item $\lambda'(u^j_B,u^j_R)=3$ (resp. $6,9$)
            \item $\lambda'(u^j_R,u^j_G)=4$ (resp. $7,10$)
            \item (time-edges into and out of $uv_B,uv_R,uv_G$ are all bold)
            \item $\lambda(v_B^i,v_R^i)=11$ (resp. $14,17$)
            \item $\lambda(v_R^i,v_G^i)=12$ (resp. $15,18$)
            \item $\lambda(v_G^i,v)=13$ (resp. $16,19$)
        \end{itemize}

        It should be clear that this labeling creates a temporal path from $u$ to $v$ for each edge $(u,v)\in G$ such that the traveler demands are satisfied (via $uv_B,uv_R,$ or $uv_G$ depending on whether the edge was colored $B,R,$ or $G$ under $P$).
        
        We now show hermit demands are satisfied as well: because $P$ is a proper 3-edge-coloring of a cubic graph, every vertex is incident to exactly one edge of each color. 
        
        In $A$-gadgets, the hermit starting at $s^u_B$ (resp. $s^u_R,s^u_G$) has a temporal path to $u^i_B$ (resp. $u^i_R,u^i_G$) arriving by time $2$ (resp. $6,10$) if the edge from $u$ to its $i$th neighbor is assigned $B$ (resp. $R,G$) under $C$. The hermit can then (if $i \ne 3$) use the bold time-edges to reach $u_B^3$ (resp. $u^3_R, u^3_G$).

        In $B$-gadgets, the hermit starting at $s^v_B$ (resp. $s^v_R,s^v_G$) has a temporal path to $v^i_B$ (resp. $v^i_R,v^i_G$) arriving by time $10$ (resp. $14,18$) using the bold time-edges. If the edge from $v$ to its $j$th neighbor is assigned $B$ (resp. $R,G$) under $C$, then the hermit can extend this path by using the spoke edges from $v_B^j$ (resp. $v_R^j,v_G^j$) into $v$.
    \end{claimproof}

    The remainder of the proof is devoted to showing the opposite implication; that is, if \delaybetter ~instance $(G',\lambda),D$ is a yes-instance  (i.e., there exists some delaying $\lambda'$ of $\lambda$ satisfying all demands in $D$) then the \marxcol ~instance $G,P$ is a yes-instance.
    For some $\lambda'$, we say that the traveler from $u$ to $v$ is \emph{blue} (resp. \emph{red, green}) if that traveler is routed through a vertex $uv_B$ (resp. $uv_R, uv_G$). (If several paths are possible, one may be chosen arbitrarily - though as we shall see this never happens.) No traveler has more than one color: each traveler goes through exactly one edge-gadget, from its starting $A$-gadget to its ending $B$-gadget (due to the bold time-edges enforcing the direction of the edge-gadget).  

    We make repeated use of the fact that, by construction, bold time-edges are never delayed. Note that if $\lambda=\bf 1$ everywhere including bold gadgets, then these force their edge to be at exactly the intended time in the delaying $\lambda'$.

    \begin{claim}\label{clm:a-gadg-times}
        Let $u\in A$. Then there is exactly one $i$ such that $\lambda'(u,u_B^i)\in[2,4]$ (resp. $[5,7],[8,10]$); there is at least one $i$ such that $\lambda'(u_B^i,u_R^i)\in[6,8]$ (resp. [9-11]); and there is at least one $i$ such that $\lambda'(u_R^i,u_G^i)\in[10,12]$.
    \end{claim}

    \begin{claimproof}
        The claim holds as a consequence of the hermit demands. The blue (resp. red, green) hermit must reach the blue (resp. red, green) layer using at least one (resp. two, three) spoke edge(s), arriving by time $4$ (resp. $8$, $12$) at the latest and departing from $u$ at time $2$ (resp. $5$, $8$) at the earliest.
    \end{claimproof}

    \begin{claim}\label{clm:a-gadg-thirds}
        At most $\frac{1}{3}$ of travelers are blue and at most $\frac{2}{3}$ of travelers are red or blue. 
    \end{claim}

    \begin{claimproof}
        First, suppose over a third of travelers are blue. Then the $A$-gadget of some vertex $u$ has at least two travelers reaching different vertices of its green layer by time $6$ (and, necessarily, different vertices of its red layer by time $5$). This entails that at least two of the spoke edges between the blue and red layers in that gadget are at time $5$ or less, which contradicts \cref{clm:a-gadg-times}.
        Similarly, if over two thirds of travelers are red or blue, then the $A$-gadget of some vertex $u$ has at least three travelers reaching three different vertices of the green layer by time $9$, entailing that the three spoke edges from the red layer to the green layer are at time $9$ or earlier and again contradicting \cref{clm:a-gadg-times}.
    \end{claimproof}

    The following result is obtained through similar reasoning to that for \cref{clm:a-gadg-times}:
    \begin{restatable}[$\ast$]{claim}{bGadgHermitsClaim}\label{clm:b-gadg-hermits}
        Let $v\in B$. Then under $\lambda'$, there is some $i$ such that $v_B^i-v_R^i-v_G^i-v$ is a temporal path with departure time in $[9,11]$ and arrival time in $[11,13]$; there is some $i$ such that $v_R^i-v_G^i-v$ is a temporal path with departure time in $[13,15]$ and arrival time in $[14,16]$; and there is some $i$ such that $\lambda'(v_G^i,v)\in[17,19]$.
    \end{restatable}
    


    \begin{claim}\label{clm:b-gadg-arrivals}
        Let $v\in B$. Then at least 1 traveler arrives at the $B$-gadget of $v$ at time $8$; and at least 2 travelers arrive at the $B$-gadget of $v$ at time $8$ or time $11$.
    \end{claim}
    \begin{claimproof}
        First note that all travelers arriving at the $B$-gadget come from some edge-gadget and consequently arrive at a time in $\{8,11,14\}$.
        Applying \cref{clm:b-gadg-hermits}, there is some $i$ such that any traveler arriving at $v_B^i$ strictly after time $10$ would be stranded there -- so the traveler arriving from the $i$th neighbor of $v$ must arrive at time $8$. Likewise, there is some $j$ different from $i$ such that any traveler arriving at $v_R^j$ strictly after time $14$ would be stranded there -- so the traveler arriving from the $j$th neighbor of $v$ must arrive at $v_R^j$ by time $14$ and so at $v_B^j$ at time $8$ or time $11$.
    \end{claimproof}

    The proof of the following is similar to that of \cref{clm:a-gadg-thirds}:
    \begin{restatable}[$\ast$]{claim}{bGadgThirdsClaim}\label{clm:b-gadg-thirds}
        At least $\frac{1}{3}$ of travelers are blue and at least $\frac{2}{3}$ of travelers are red or blue.
    \end{restatable}


    For some $\lambda'$, we say that the traveler from $u$ to $v$ is \emph{blue} (resp. \emph{red, green}) if the temporal path used to route that traveler goes through a vertex $uv_B$ (resp. $uv_R, uv_G$). 

    \begin{claim}
        The colors of travelers in $(G',\lambda')$ fully specify a proper edge-coloring of $G$ which is consistent with the precoloring $P$. 
    \end{claim}
    \begin{claimproof}
        First, note that the precoloring is consistent with $P$ because precolored edges in $G$ have edge-gadgets consisting of only one vertex, ensuring that the traveler is assigned the appropriate color.

        Next, observe that Claims \ref{clm:a-gadg-thirds} and \ref{clm:b-gadg-thirds} together entail that \emph{exactly} $\frac{1}{3}$ of travelers are blue and \emph{exactly} $\frac{1}{3}$ of travelers are red. Moreover, the proof of those claims holds locally; exactly one of the three travelers leaving any given $A$-vertex is blue (resp. red), and exactly one of the three travelers arriving at any given $B$-vertex is blue (resp. red).
    \end{claimproof}

    This concludes the proof that the \marxcol ~instance $(G,P)$ is a yes-instance if the \delaybetter ~instance $(G,\lambda),D$ was a yes-instance. 
    
    We emphasize at this point that our construction preserves planarity and that in the directed case, the footprint contains no directed cycles. We recall that in the undirected case bold time-edges enforce that travelers can only go from an $A$-gadget to a $B$-gadget once.
    The maximum degree in the graph is $10$ (due to vertices $u_G^i$, which are incident to 4 bold gadgets in addition to 6 normal edges). 
    Note that the proof still holds if the initial temporal assignment $\lambda$ assigns time $2$ to every non-bold edge in an $A$-gadget and time $9$ to every non-bold edge in a $B$-gadget, in which case the largest delay is of $10$ (delaying a time-edge from the green layer of a $B$-gadget to a $B$-vertex $v$ to be at time $19$). 
    Consequently, our proof also shows that $\delta$-\delaybetter ~is NP-hard for $\delta\ge 9$. 
    
    On the other hand, the proof also holds if the initial temporal assignment is instead the constant function $\bf 1$: studying \cref{fig:bold-time-edges} it can be seen that this would still result in bold time-edges being assigned the intended time under $\lambda'$. 

    We have membership of NP from \cref{lem:in-np}, and the result follows.
\end{proof}

\bibliography{sas-bib2doi}

\appendix

\section{Deferred proofs}
\subsection{Reduction from \texorpdfstring{$\delta$}{delta}-DB to DB}

\earlyHermitsClaim*

\begin{claimproof}
    The claim follows straightforwardly from pareto optimality of $\lambda^*$, and the fact that hermit trailheads are used only by the hermit, who can wait at the vertex $u$ instead of at $u'$.
\end{claimproof}


    


\maxDelayTwoDeltaClaim*

    \begin{claimproof}
        By construction, there is a hermit demand $(u',v',2t+2\delta+1)$. This hermit must use the edge $(u,uv)$ (since the new vertex $uv$ has no other incoming edges and $v'$ is only reachable from $uv$). The hermit must use this edge no earlier than time $2t$
        (as this is its original time under $\lambda'$)
        and no later than time $2t+2\delta$
        (as the next edge in the temporal path must be at time $2t+2\delta+1$ exactly).
    \end{claimproof}

\subsection{Final lifetime is polynomial \wlogsm}

\polyTmaxLemma*
\begin{proof}
Given an instance $(\tgraph, D)$ of either problem, we identify the set of all \emph{explicit} times (directly encoded in the input) as $T_{explicit} := \{d_t | d \in D\} \cup \{\lambda(e) | e \in E(\tgraph)\}
$, where $d_t$ is the arrival time specified by $d$.
Denote $|T_{explicit}|$ by $\alpha \le |E| + |D|$ (this inequality is strict if any time appears explicitly more than once in $(\tgraph,D)$). 
We then may sort $T_{explicit}$ into an ordered list of times $t_1 < t_2 < \ldots < t_{\alpha}$. 

\emph{Shrinking} of an interval $[t_i,t_j]$ to be of size $\ell$ consists in decrementing all times $t_j$ or greater in the original instance by $t_j - \ell- t_i \ge 0$. Thus, any edge (or demand) formerly at time $t_j$ is updated to be at time $t_i+\ell$. \emph{Deleting} a time interval $[t_i,t_j]$ consists in shrinking that time interval to have size $0$.

We first deal with \textsc{DelayBetter} and \textsc{Path DB}. 
Consider the integer intervals $[t_i, t_{i+1}]$. If any such interval has size greater than $|E|$, we may without loss shrink the interval to have size $|E|$ instead. 
No-instances of both problems are clearly preserved by the operation. Yes-instances are also preserved: only the relative order of times assigned to edges matters for a temporal path to exists, and any ordering achievable in the original instance is also achievable in the transformed instance since at most $|E|$ unique times are assigned under $\lambda'$ in total. 

We now deal with \textsc{$\delta$-DelayBetter}. We identify the set of \emph{relevant} times to be $T_{relevant}:= \bigcup_{t \in T_{explicit}}: [t,t+\delta]$. Note that this set has cardinality at most $\delta \cdot (|E| + |D|)$, and that it contains all possible times used in any solution $\lambda'$. Hence we then may eliminate every time not in $T_{relevant}$ (by deleting at most $|E|+|D|$ intervals) and obtain an equisatisfiable instance with $T_{\max} \le \delta \cdot \alpha$. 

In both cases, the procedure clearly runs in time $\mathrm{poly}(\log T_{\max} + |V(\tgraph)| + |D|)$, and we obtain the desired result. 
\end{proof}

\subsection{Containment in NP}

\inNPLemma*

\begin{proof}
    Given an instance $I=(\tgraph, D)$ of \textsc{($\delta$-)DelayBetter} and a corresponding solution, i.e., an assignment $\lambda'$ of time-labels (which can delay edges of the initial assignment $\lambda$), we can check in polynomial time whether $\lambda'$ is indeed a valid solution for $I$ as follows.

    First, we need to check that the assignment $\lambda'$ actually represents valid delays (i.e., that no edge was moved to an earlier point in time). 
    To do so, we check in $O(|\mathcal{E}|)$ whether for every $e \in \mathcal{E}$ we have $\lambda(e) \leq \lambda'(e)$ (for the case of \textsc{$\delta$-DelayBetter}, we also check that $\lambda(e) + \delta \le \lambda'(e)$).

    It remains to check the demands are met by the assignment. 
    The earliest arrival time $\mathrm{arr}_{u\to v}$ of any strict temporal path from $u$ to $v$ in the temporal graph $(G, \lambda')$ may be computed in polynomial time (see, e.g. \cite{wu_path_2014}). It then suffices to verify, for each $(u,v,t)\in D$, that $\mathrm{arr}_{u\to v}\le t$, which can be done in polynomial time, and the result follows.
 \end{proof}

\subsection{Integrality of the Linear Program}
We restate \cref{clm:integral-lp}. For convenience, we also include the LP again here:
\begin{align*}
    &\mathrm{maximize} \sum_{d\in D} d_t - t_{d_f}' \text{ subject to}\\
	&t_{uv} \le t'_{uv} \text{ for each } (u,v)\in E(G)\\
	&t'_{uv} \le  t_{uv} + \delta \text{ for each } (u,v)\in E(G) \text{ (only for $\delta$-\textsc{Path-DB})}\\
	&t_{uv}' \le t'_{vw}+1 \text{ for each pair of edges $uv$ and $vw$ such that $uv \prec vw$}\\
	&t_{d_f}' \le d_t\text{ for each demand $d$}
\end{align*}

\integralLPClaim* 

\begin{claimproof}

    Suppose otherwise. That is, there is some non-integral solution $X$ to the LP which is strictly better than any integral solution.
    
    Under $X$, for some edge $vw$, $t_{vw}'$ is assigned a non-integer value, say $x= y + \epsilon$ with $y \in \mathbb{N}$ and $0<\epsilon < 1$.

    Consider the assignment obtained by instead setting $t_{vw}'=y$. If this is still a valid solution to the LP, then this clearly does not worsen the objective (and cannot improve it since we assumed $X$ was optimal). Apply this update iteratively, everywhere possible, and consider the new solution $Y$ obtained. By our initial premise, $Y$ is still not an integral solution, and by construction $Y$ has the same objective value as $X$ and also would cease to be a solution if any of its non-integer variables were rounded down to the nearest integer.

    We again can find some (possibly different) edge $vw$ such that $t_{vw}'$ is assigned a non-integer value under $Y$, now $y = z + \epsilon$ with $z \in \mathbb{N}$ and $0<\epsilon < 1$. 

    Consider the assignment obtained by instead setting $t_{vw}'=z$. Necessarily this assignment is not a valid solution for the LP (since otherwise we already would have performed the update). Consequently, there is some constraint which is violated by the update, which necessarily has form $t_{uv}' \le t'_{vw} + 1$, since all other types of constraints would remain satisfied if we set $t_{uv}'=z$. Moreover, $t_{uv}'$ must itself be assigned some non-integer value (strictly less than that assigned to $t_{vw}'$) under $Y$. By iteratively applying the same logic (and the fact that there are only finitely many edges) we conclude some edge must be assigned a non-integer value under $Y$ even though it could have been rounded down to the nearest integer - contradicting a central property of the assignment $Y$. 
    We note that our construction for $Y$ may be performed in polynomial time to iteratively construct an integral solution from a non-integral one, and the claim follows.
\end{claimproof}

\subsection{Hardness for planar instances}

\bGadgHermitsClaim*

\begin{claimproof}
    Analogously to the proof of \cref{clm:a-gadg-times}, we need only concern ourselves with hermits to prove this claim. 
    The blue hermit must travel from the blue layer to $v$ as specified in the claim (since it cannot make use of any bold edges outside the blue layer in the temporal path). Similarly, the red hermit must reach $v$ by a temporal path not using any bold edges in the green layer, and the green hermit must reach $v$ using some spoke edge from the green layer in the interval $[17,19]$.
\end{claimproof}

\bGadgThirdsClaim*
\begin{claimproof}
        The proof is similar to that of \cref{clm:a-gadg-thirds}. If less than a third of travelers are blue, then some $B$-gadget has all three travelers arriving strictly after time $8$, contradicting \cref{clm:b-gadg-arrivals}. And if less than two thirds of travelers are red or blue, then some $B$-gadget has at least two travelers arriving at time $14$, again contradicting \cref{clm:b-gadg-arrivals}.        
    \end{claimproof}

\end{document}

%% file: preamble.tex
\usepackage[T1]{fontenc}
\usepackage{graphicx}
\usepackage{amsfonts,amssymb}
\usepackage[ruled,vlined,linesnumbered]{algorithm2e}
\usepackage{xcolor}

\usepackage{hyperref}

\usepackage{hhline}
\usepackage{multicol, multirow}
\usepackage{array}
\usepackage{tabularray}
\usepackage{framed}
\usepackage{csquotes}
\usepackage{pdfpages}
\usepackage{needspace}

\newcommand{\fes}{\rho}
\newcommand{\tgraph}{\mathcal{G}}
\newcommand{\footprint}{\mathcal{G}_\downarrow}
\newcommand{\NN}{\mathbb{N}}
\newcommand{\wlogsm}{without loss of generality }

\newcommand{\delaybetter}{\textsc{DelayBetter}}
\newcommand{\marxcoloring}{\textsc
    {Cubic Bipartite Planar Edge Precoloring Extension}}
\newcommand{\marxcol}{\textsc{CBP-EPE}}
\newcommand{\tinit}{T_{\mathrm{init}}}
\newcommand{\tmax}{T_{\max}}

\usepackage{lmodern} 
\usepackage{bm}

\AtBeginDocument{%
    \DeclareFontShape{T1}{lmr}{m}{scit}{<->ssub*lmr/m/scsl}{}%
    \DeclareFontShape{T1}{lmr}{bx}{sc}{<->ssub*lmr/m/sc}{}%
}

%% file: SAS-SAND.bbl
\begin{thebibliography}{10}

\bibitem{binder_multi-objective_2017}
Stefan Binder, Yousef Maknoon, and Michel Bierlaire.
\newblock The multi-objective railway timetable rescheduling problem.
\newblock {\em Transportation Research Part C: Emerging Technologies}, 78:78--94, may 2017.
\newblock \href {https://doi.org/10.1016/j.trc.2017.02.001} {\path{doi:10.1016/j.trc.2017.02.001}}.

\bibitem{bumpus_vimw_2021}
Benjamin~Merlin Bumpus and Kitty Meeks.
\newblock Edge exploration of temporal graphs.
\newblock In Paola Flocchini and Lucia Moura, editors, {\em Combinatorial Algorithms}, Lecture Notes in Computer Science, pages 107--121, Cham, 2021. Springer.
\newblock \href {https://doi.org/10.1007/978-3-030-79987-8_8} {\path{doi:10.1007/978-3-030-79987-8_8}}.

\bibitem{cacchiani_overview_2014}
Valentina Cacchiani, Dennis Huisman, Martin Kidd, Leo Kroon, Paolo Toth, Lucas Veelenturf, and Joris Wagenaar.
\newblock An overview of recovery models and algorithms for real-time railway rescheduling.
\newblock {\em Transportation Research Part B: Methodological}, 63:15--37, may 2014.
\newblock \href {https://doi.org/10.1016/j.trb.2014.01.009} {\path{doi:10.1016/j.trb.2014.01.009}}.

\bibitem{carosi_delay_2015}
S.~Carosi, S.~Gualandi, F.~Malucelli, and E.~Tresoldi.
\newblock Delay {Management} in {Public} {Transportation}: {Service} {Regularity} {Issues} and {Crew} {Re}-scheduling.
\newblock {\em Transportation Research Procedia}, 10:483--492, 2015.
\newblock \href {https://doi.org/10.1016/j.trpro.2015.09.002} {\path{doi:10.1016/j.trpro.2015.09.002}}.

\bibitem{casteigts_simple_2022}
Arnaud Casteigts, Timothée Corsini, and Writika Sarkar.
\newblock Simple, strict, proper, happy: {A} study of reachability in temporal graphs, aug 2022.
\newblock URL: \url{http://arxiv.org/abs/2208.01720}, \href {https://doi.org/10.48550/arXiv.2208.01720} {\path{doi:10.48550/arXiv.2208.01720}}.

\bibitem{deligkas_minimizing_2023}
Argyrios Deligkas, Eduard Eiben, and George Skretas.
\newblock Minimizing {Reachability} {Times} on {Temporal} {Graphs} via {Shifting} {Labels}.
\newblock In {\em Proceedings of the {Thirty}-{Second} {International} {Joint} {Conference} on {Artificial} {Intelligence}}, pages 5333--5340, Macau, SAR China, aug 2023. International Joint Conferences on Artificial Intelligence Organization.
\newblock \href {https://doi.org/10.24963/ijcai.2023/592} {\path{doi:10.24963/ijcai.2023/592}}.

\bibitem{deligkas_optimizing_2022}
Argyrios Deligkas and Igor Potapov.
\newblock Optimizing reachability sets in temporal graphs by delaying.
\newblock {\em Information and Computation}, 285:104890, may 2022.
\newblock \href {https://doi.org/10.1016/j.ic.2022.104890} {\path{doi:10.1016/j.ic.2022.104890}}.

\bibitem{deutschebahn_punctuality_2024}
{P}unctuality | {D}eutsche {B}ahn {I}nterim {R}eport 2024.
\newblock \url{https://zbir.deutschebahn.com/2024/en/interim-group-management-report-unaudited/product-quality-and-digitalization/punctuality/}.
\newblock [Accessed 19-09-2024].

\bibitem{dollevoet_delay_2012}
Twan Dollevoet, Dennis Huisman, Marie Schmidt, and Anita Schöbel.
\newblock Delay {Management} with {Rerouting} of {Passengers}.
\newblock {\em Transportation Science}, 46(1):74--89, feb 2012.
\newblock \href {https://doi.org/10.1287/trsc.1110.0375} {\path{doi:10.1287/trsc.1110.0375}}.

\bibitem{enright_reachability_2025}
Jessica~A. Enright, Laura Larios{-}Jones, Kitty Meeks, and William Pettersson.
\newblock Reachability in temporal graphs under perturbation.
\newblock In {\em {SOFSEM} 2025: Theory and Practice of Computer Science - 50th International Conference on Current Trends in Theory and Practice of Computer Science, {SOFSEM} 2025, Bratislava, Slovak Republic, January 20-23, 2025, Proceedings, Part {I}}, volume 15538 of {\em Lecture Notes in Computer Science}, pages 255--269. Springer, 2025.
\newblock \href {https://doi.org/10.1007/978-3-031-82670-2_19} {\path{doi:10.1007/978-3-031-82670-2_19}}.

\bibitem{antunes2019characterizing}
Ivan Tadeu Ferreira~Antunes Filho.
\newblock {\em Characterizing Boolean satisfiability variants}.
\newblock PhD thesis, Massachusetts Institute of Technology, 2019.
\newblock URL: \url{https://dspace.mit.edu/handle/1721.1/124241}.

\bibitem{fuchsle_delay-robust_2022}
Eugen Füchsle, Hendrik Molter, Rolf Niedermeier, and Malte Renken.
\newblock Delay-{Robust} {Routes} in {Temporal} {Graphs}, jan 2022.
\newblock arXiv:2201.05390 [cs].
\newblock URL: \url{http://arxiv.org/abs/2201.05390}, \href {https://doi.org/10.48550/arXiv.2201.05390} {\path{doi:10.48550/arXiv.2201.05390}}.

\bibitem{fuchsle_temporal_2022}
Eugen Füchsle, Hendrik Molter, Rolf Niedermeier, and Malte Renken.
\newblock Temporal {Connectivity}: {Coping} with {Foreseen} and {Unforeseen} {Delays}, jan 2022.
\newblock arXiv:2201.05011 [cs].
\newblock URL: \url{http://arxiv.org/abs/2201.05011}, \href {https://doi.org/10.48550/arXiv.2201.05011} {\path{doi:10.48550/arXiv.2201.05011}}.

\bibitem{kanade_railway_2004}
Michael Gatto, Björn Glaus, Riko Jacob, Leon Peeters, and Peter Widmayer.
\newblock Railway {Delay} {Management}: {Exploring} {Its} {Algorithmic} {Complexity}.
\newblock In {\em Algorithm {Theory} - {SWAT} 2004}, volume 3111, pages 199--211. Springer Berlin Heidelberg, Berlin, Heidelberg, 2004.
\newblock Series Title: Lecture Notes in Computer Science.
\newblock \href {https://doi.org/10.1007/978-3-540-27810-8_18} {\path{doi:10.1007/978-3-540-27810-8_18}}.

\bibitem{hutchison_computational_2005}
Michael Gatto, Riko Jacob, Leon Peeters, and Anita Schöbel.
\newblock The {Computational} {Complexity} of {Delay} {Management}.
\newblock In {\em Graph-{Theoretic} {Concepts} in {Computer} {Science}}, volume 3787, pages 227--238. Springer Berlin Heidelberg, Berlin, Heidelberg, 2005.
\newblock Series Title: Lecture Notes in Computer Science.
\newblock \href {https://doi.org/10.1007/11604686_20} {\path{doi:10.1007/11604686_20}}.

\bibitem{ginkel_wait_2007}
Andreas Ginkel and Anita Schöbel.
\newblock To {Wait} or {Not} to {Wait}? {The} {Bicriteria} {Delay} {Management} {Problem} in {Public} {Transportation}.
\newblock {\em Transportation Science}, 41(4):527--538, nov 2007.
\newblock \href {https://doi.org/10.1287/trsc.1070.0212} {\path{doi:10.1287/trsc.1070.0212}}.

\bibitem{heilporn_optimization_2008}
Géraldine Heilporn, Luigi De~Giovanni, and Martine Labbé.
\newblock Optimization models for the single delay management problem in public transportation.
\newblock {\em European Journal of Operational Research}, 189(3):762--774, sep 2008.
\newblock \href {https://doi.org/10.1016/j.ejor.2006.10.065} {\path{doi:10.1016/j.ejor.2006.10.065}}.

\bibitem{kahn_topological_1962}
A.~B. Kahn.
\newblock Topological sorting of large networks.
\newblock {\em Communications of the ACM}, 5(11):558--562, nov 1962.
\newblock \href {https://doi.org/10.1145/368996.369025} {\path{doi:10.1145/368996.369025}}.

\bibitem{karmarkay1984linearprog}
N.~Karmarkar.
\newblock A new polynomial-time algorithm for linear programming.
\newblock In {\em Proceedings of the Sixteenth Annual ACM Symposium on Theory of Computing}, STOC '84, page 302–311, New York, NY, USA, 1984. Association for Computing Machinery.
\newblock \href {https://doi.org/10.1145/800057.808695} {\path{doi:10.1145/800057.808695}}.

\bibitem{geraets_cyclic_2007}
Leo~G. Kroon, Rommert Dekker, and Michiel J. C.~M. Vromans.
\newblock Cyclic {Railway} {Timetabling}: {A} {Stochastic} {Optimization} {Approach}.
\newblock In Frank Geraets, Leo Kroon, Anita Schoebel, Dorothea Wagner, and Christos~D. Zaroliagis, editors, {\em Algorithmic {Methods} for {Railway} {Optimization}}, volume 4359, pages 41--66. Springer Berlin Heidelberg, Berlin, Heidelberg, 2007.
\newblock Series Title: Lecture Notes in Computer Science.
\newblock \href {https://doi.org/10.1007/978-3-540-74247-0_2} {\path{doi:10.1007/978-3-540-74247-0_2}}.

\bibitem{kutner_temporal_2024}
David~C. Kutner and Laura Larios-Jones.
\newblock Temporal {Reachability} {Dominating} {Sets}: contagion in temporal graphs, may 2024.
\newblock arXiv:2306.06999 [cs, math].
\newblock URL: \url{http://arxiv.org/abs/2306.06999}, \href {https://doi.org/10.48550/arXiv.2306.06999} {\path{doi:10.48550/arXiv.2306.06999}}.

\bibitem{konig_review_2020}
Eva König.
\newblock A review on railway delay management.
\newblock {\em Public Transport}, 12(2):335--361, jun 2020.
\newblock \href {https://doi.org/10.1007/s12469-020-00233-1} {\path{doi:10.1007/s12469-020-00233-1}}.

\bibitem{malucelli_delay_2019}
Federico Malucelli and Emanuele Tresoldi.
\newblock Delay and disruption management in local public transportation via real-time vehicle and crew re-scheduling: a case study.
\newblock {\em Public Transport}, 11(1):1--25, jun 2019.
\newblock \href {https://doi.org/10.1007/s12469-019-00196-y} {\path{doi:10.1007/s12469-019-00196-y}}.

\bibitem{marx_np-completeness_2005}
Dániel Marx.
\newblock {NP}-completeness of list coloring and precoloring extension on the edges of planar graphs.
\newblock {\em Journal of Graph Theory}, 49(4):313--324, 2005.
\newblock \_eprint: https://onlinelibrary.wiley.com/doi/pdf/10.1002/jgt.20085.
\newblock \href {https://doi.org/10.1002/jgt.20085} {\path{doi:10.1002/jgt.20085}}.

\bibitem{molter_temporal_2024}
Hendrik Molter, Malte Renken, and Philipp Zschoche.
\newblock Temporal reachability minimization: {Delaying} vs. deleting.
\newblock {\em Journal of Computer and System Sciences}, 144:103549, sep 2024.
\newblock \href {https://doi.org/10.1016/j.jcss.2024.103549} {\path{doi:10.1016/j.jcss.2024.103549}}.

\bibitem{moret_planarNAE3SAT_1988}
B.~M.~E. Moret.
\newblock {Planar NAE3SAT is in P}.
\newblock {\em SIGACT News}, 19(2):51–54, jun 1988.
\newblock \href {https://doi.org/10.1145/49097.49099} {\path{doi:10.1145/49097.49099}}.

\bibitem{schachtebeck_delay_2010}
Michael Schachtebeck.
\newblock {\em Delay {Management} in {Public} {Transportation}: {Capacities}, {Robustness}, and {Integration}}.
\newblock PhD thesis, Georg-August-University Göttingen, 2010.
\newblock \href {https://doi.org/10.53846/goediss-2538} {\path{doi:10.53846/goediss-2538}}.

\bibitem{schobel_model_2001}
Anita Schöbel.
\newblock A {Model} for the {Delay} {Management} {Problem} based on {Mixed}-{Integer}-{Programming}.
\newblock {\em Electronic Notes in Theoretical Computer Science}, 50(1):1--10, aug 2001.
\newblock \href {https://doi.org/10.1016/S1571-0661(04)00160-4} {\path{doi:10.1016/S1571-0661(04)00160-4}}.

\bibitem{schobel_anita_optimization_2006}
{Schöbel, Anita}.
\newblock {\em Optimization in {Public} {Transportation}}, volume~3 of {\em Springer {Optimization} and {Its} {Applications}}.
\newblock Springer US, Boston, MA, 2006.
\newblock \href {https://doi.org/10.1007/978-0-387-36643-2} {\path{doi:10.1007/978-0-387-36643-2}}.

\bibitem{scozzaro_optimizing_2023}
Geoffrey Scozzaro, Clara Buire, Daniel Delahaye, and Aude Marzuoli.
\newblock Optimizing air-rail travel connections: A data-driven delay management strategy for seamless passenger journeys.
\newblock In {\em SESAR Innovation Days}, 2023.

\bibitem{veelenturf_railway_2016}
Lucas~P. Veelenturf, Martin~P. Kidd, Valentina Cacchiani, Leo~G. Kroon, and Paolo Toth.
\newblock A {Railway} {Timetable} {Rescheduling} {Approach} for {Handling} {Large}-{Scale} {Disruptions}.
\newblock {\em Transportation Science}, 50(3):841--862, aug 2016.
\newblock \href {https://doi.org/10.1287/trsc.2015.0618} {\path{doi:10.1287/trsc.2015.0618}}.

\bibitem{wu_path_2014}
Huanhuan Wu, James Cheng, Silu Huang, Yiping Ke, Yi~Lu, and Yanyan Xu.
\newblock Path problems in temporal graphs.
\newblock {\em Proc. VLDB Endow.}, 7(9):721–732, may 2014.
\newblock \href {https://doi.org/10.14778/2732939.2732945} {\path{doi:10.14778/2732939.2732945}}.

\bibitem{zhang_train_2023}
Chuntian Zhang, Yuan Gao, Valentina Cacchiani, Lixing Yang, and Ziyou Gao.
\newblock Train rescheduling for large-scale disruptions in a large-scale railway network.
\newblock {\em Transportation Research Part B: Methodological}, 174:102786, aug 2023.
\newblock \href {https://doi.org/10.1016/j.trb.2023.102786} {\path{doi:10.1016/j.trb.2023.102786}}.

\bibitem{zhu_integrated_2020}
Yongqiu Zhu and Rob M.~P. Goverde.
\newblock Integrated timetable rescheduling and passenger reassignment during railway disruptions.
\newblock {\em Transportation Research Part B: Methodological}, 140:282--314, oct 2020.
\newblock \href {https://doi.org/10.1016/j.trb.2020.09.001} {\path{doi:10.1016/j.trb.2020.09.001}}.

\end{thebibliography}
